\documentclass[11pt,a4paper]{article}

\usepackage{authblk}

\usepackage{amsmath}
\usepackage{amssymb}
\usepackage{amsthm}
\usepackage{mathabx}
\usepackage{fullpage}
\usepackage{mdwlist}
\usepackage{hyperref}
\usepackage{graphicx}
\usepackage{bbm}
\usepackage{bm}

\usepackage{tikz}
\usetikzlibrary{decorations.pathmorphing}

\newtheorem{theorem}{Theorem}
\newtheorem{corollary}[theorem]{Corollary}
\newtheorem{lemma}[theorem]{Lemma}

\newtheorem{proposition}[theorem]{Proposition}

\newtheorem{remark}[theorem]{Remark}

\newcommand{\poly}{\mathop{\mathrm{poly}}}

\newcommand{\F}{\mathbb{F}}

\newcommand{\prefix}{\sqsupset}

\title{Towards Optimal Depth-Reductions for Algebraic Formulas}
\date{}
\author[1]{Herv\'{e} Fournier}
\author[2]{Nutan Limaye}
\author[3]{Guillaume Malod}
\author[4]{Srikanth Srinivasan\footnote{The author is grateful for a research visit sponsored by the Guest researchers faculty program at Universit\'{e} Paris Cit\'{e} in summer 2022.}}
\author[5]{S\'{e}bastien Tavenas}
\affil[1,3]{Universit\'{e} Paris Cit\'{e}, IMJ-PRG \protect\\
\textit{Emails: herve.fournier@imj-prg.fr, guillaume.malod@imj-prg.fr}}
\affil[2]{ITU Copenhagen, \protect\\ \textit{Email: nuli@itu.dk}}
\affil[5]{Aarhus University, \protect\\ \textit{Email: srinivasan.srikanth@gmail.com}}
\affil[3]{Univ. Savoie Mont Blanc, CNRS, LAMA \protect\\ \textit {Email: sebastien.tavenas@univ-smb.fr}}

\begin{document}
\maketitle
\thispagestyle{empty}

\begin{abstract}
Classical results of Brent, Kuck and Maruyama (IEEE Trans.~Computers 1973) and Brent (JACM 1974) show that any algebraic formula of size~$s$ can be converted to one of depth~$O(\log s)$ with only a polynomial blow-up in size. In this paper, we consider a fine-grained version of this result depending on the degree of the polynomial computed by the algebraic formula.

Given a homogeneous algebraic formula of size~$s$ computing a polynomial $P$ of degree~$d$, we show that $P$ can also be computed by an (unbounded fan-in) algebraic formula of depth~$O(\log d)$ and size~$\poly(s).$ Our proof shows that this result also holds in the highly restricted setting of monotone, non-commutative algebraic formulas. 

This improves on previous results in the regime when $d$ is small (i.e.,~$d = s^{o(1)}$). In particular, for the setting of $d = O(\log s),$ along with a result of Raz (STOC 2010, JACM 2013), our result implies the same depth reduction even for \emph{inhomogeneous} formulas. This is particularly interesting in light of recent algebraic formula lower bounds, which work precisely in this ``low-degree'' and ``low-depth'' setting.

We also show that these results cannot be improved in the monotone setting, even for commutative formulas.
\end{abstract}

\newpage

\section{Introduction}
\label{sec:intro}

In this paper, we study a basic question regarding computational tradeoffs between two resources for the model of \emph{algebraic formulas}.

An algebraic formula $F$ for a multivariate polynomial $P(x_1,\ldots, x_n)$ is simply an algebraic expression for $P$ made up of nested additions and multiplications. Equivalently, it can be defined as a rooted directed tree where the leaves are labelled by variables and internal nodes (or gates) compute either linear combinations or products of their children (a formal definition can be found in Section~\ref{sec:prelims} below). Unless otherwise stated, we do not bound the number of children of a gate (in other words, we consider formulas of \emph{unbounded fan-in}).

The two basic computational resources that describe the complexity of an algebraic formula $F$ are its \emph{size}, which is the number of leaves in the underlying tree, and its \emph{depth}, which naturally is the depth of the tree. Polynomials\footnote{Strictly speaking, we should refer here to infinite sequences of polynomials, but we ignore this distinction.} that have efficient (i.e.,~$\poly(n)$-sized) algebraic formulas form the algebraic complexity class $\mathrm{VF}.$ Like its Boolean counterpart $\mathrm{NC}^1$, this is a natural and important complexity class. 

Tradeoffs between size and depth in the setting of formulas and related models of computation have been the focus of many previous works, starting from the early 1970s~\cite{Spira71tradeoffs,BKM, Brent, VSBR, ShamirSnir, Nisan, Raznc2nc1, RYbalance, AV, Koiran, Tav13, GKKSdepth3, KSS, CKSV, LST-ECCC}. We describe a few such results here.

\begin{itemize}
\item In the Boolean setting, Spira~\cite{Spira71tradeoffs} and independently Khrapchenko (see~\cite{YK68}) showed that any Boolean formula of size $s$ can be converted to a Boolean formula of depth $O(\log s)$ while keeping the size bounded by~$s^{O(1)}.$ These results were replicated in the algebraic setting in results of Brent, Kuck and Marayuma~\cite{BKM} and Brent~\cite{Brent}. The constants involved in the bounds for the depth and the size were improved in many follow-up works~\cite{PreparataMuller1, PreparataMuller2, Kos86, BCE, BonetBuss}.

\item This question has also been studied for the more general model of \emph{algebraic circuits,} where the underlying tree is replaced by a directed acyclic graph (DAG). A well-known result of Valiant, Skyum, Berkowitz and Rackoff~\cite{VSBR} showed that an algebraic circuit of size $s$ computing a polynomial of degree $\poly(s)$ can be converted to a circuit of depth\footnote{In the bounded fan-in case, this would be depth $O(\log^2 s)$ instead.} $O(\log s)$ and size $\poly(s)$. These results were also shown to hold for \emph{multilinear}\footnote{A circuit or formula is multilinear if each of its gates computes a multilinear polynomial.} circuits by Raz and Yehudayoff~\cite{RYbalance}.

\item The above results are known to be tight in various settings. In the \emph{monotone} case\footnote{where the underlying field is $\mathbb{R}$ and all constants are non-negative, so cancellations do not occur.}, Shamir and Snir~\cite{ShamirSnir} showed the existence of an explicit polynomial $P$ with a $\poly(n)$-sized circuit such that any circuit of depth $o(\log n)$ for $P$ is of superpolynomial size. Similar results were obtained in the multilinear case by  Raz~\cite{Raznc2nc1} and Chillara, Limaye and Srinivasan~\cite{CLS} (for multilinear circuits and formulas, respectively).

\item Beginning with the work of Agrawal and Vinay~\cite{AV}, a recent line of work~\cite{Koiran, Tav13, GKKSdepth3, KSS, CKSV, KOS19} has shown that algebraic circuits and formulas can be converted to formulas of \emph{constant depth} with a \emph{sub-exponential} blow-up in size. In contrast, our focus in this paper is primarily on reducing depth as much as possible while keeping the size bounded by $\poly(s)$, as in the results listed previously.
\end{itemize}

\paragraph{The question.} In this paper, we ask the question of whether stronger \emph{depth-reduction} results can be proved given a bound $d$ on the degree of the polynomial $P(x_1,\ldots,x_n)$ computed by the algebraic formula. In general, an algebraic formula of size $s$ can compute a polynomial of degree at most $s$. When $d = s$ (or $d = s^{\Omega(1)}$), the above results imply that an algebraic formula $F$ for $P$ can be converted to another formula $F'$ of depth $O(\log d)$ without significant blow-up in size. Does such a result hold for any $d$ (or more specifically, when $d = s^{o(1)}$)?

Note that this question only makes sense for algebraic formulas of unbounded fan-in. If the fan-in of each gate is bounded by a constant, then any formula of size $s$ must have depth at least $\Omega(\log s)$ (and so, Brent, Kuck, and Marayuma's result~\cite{BKM} is optimal). However, in many settings (see e.g. the third motivation below), we want a finer analysis of the formula depth that can be achieved by formulas of unbounded fan-in. 

\paragraph{Motivation.} While the question is fairly natural in our opinion, there are also many concrete reasons that lead to this line of inquiry.
\begin{itemize}
\item It is easy to see from the proof of Valiant, Skyum, Berkowitz, and Rackoff~\cite{VSBR} that any algebraic \emph{circuit} can be depth-reduced to depth $O(\log d)$ with only a $\poly(s,d)$ blow-up in size. So the natural generalization for small degrees is indeed true in the setting of algebraic circuits.
\item It also follows from previous results that the depth bound of $O(\log s)$ can be improved when the degree is $d = o(\log s).$ This is due to a result of Raz~\cite{Raztensor}: any formula $F$ for a homogeneous polynomial $P$ of degree $d$ can be converted to a \emph{homogeneous} formula\footnote{A homogeneous formula is one where each gate computes a homogeneous polynomial. This means that the formula does not compute intermediate polynomials of  degree larger than $d$.} $F'$ efficiently. Further, it is easy to see that any homogeneous formula computing a polynomial of degree $d$ has depth $O(d) = o(\log s).$ So, in this regime for the degree, standard depth-reduction results \emph{can} be strengthened.
\item Finally, very recent  results in algebraic complexity~\cite{LST-FOCS} have suggested a way of proving lower bounds against \emph{low-depth} algebraic formulas for computing \emph{low-degree} polynomials, which naturally raises the question of obtaining the best possible depth-reduction results in this setting. 

More specifically, Limaye, Srinivasan, and Tavenas~\cite{LST-FOCS} showed how to prove lower bounds against algebraic formulas (and even circuits) of small depth. Their proof proceeds by converting an algebraic formula of size $s$ and depth $\Delta$ to a homogeneous algebraic formula of size $\poly(s)$ and depth $O(\Delta),$ and then proving lower bounds against homogeneous algebraic formulas of depth $O(\Delta)$. An important point regarding the first step is that it only works in the ``low-degree setting'' of $d = O(\log s/\log \log s).$ The second step proves lower bounds against homogeneous formulas of depth up to $O(\log \log d).$

To make this proof idea work for general (unbounded-depth) algebraic formulas, we would like to be able to homogenize and depth-reduce algebraic formulas as much as possible. The aforementioned result of Raz~\cite{Raztensor} already shows that we can homogenize algebraic formulas efficiently in the low-degree setting. So it is natural to investigate the best possible depth-reduction for homogeneous algebraic formulas in the low-degree setting. 
\end{itemize}

\paragraph{Results.} Our main result is a depth-reduction result for {\em homogeneous} formulas that efficiently reduces the depth to $O(\log d)$, matching what was already known for algebraic circuits by the result of~\cite{VSBR}. 

\begin{theorem}[Main Result]
\label{thm:main-intro}
Let $F$ be a homogeneous algebraic formula of size $s$ computing a polynomial $P$ of degree $d \geq 2$. Then $P$ is also computed by a homogeneous formula $F'$ of size $\poly(s)$ and depth $O(\log d).$ Moreover, if $F$ is monotone and/or non-commutative, then so is $F'$.
\end{theorem}

Here, a \emph{monotone} algebraic formula is one that does not exploit cancellations in any way, and a \emph{non-commutative} formula describes a polynomial expression in a domain where the input variables do not commute when multiplied with each other (formal definitions are given in Section~\ref{sec:prelims}). These are both settings in which formula upper bounds are harder to prove, and hence the depth-reduction result in this setting implies the result in the standard setting.
It can also be checked that the depth-reduction procedure above preserves other interesting properties of the formula, such as \emph{multilinearity} and \emph{set-multilinearity}.

Using the aforementioned result of Raz that allows us to homogenize algebraic formulas in the low-degree setting, we get the following depth-reduction even for \emph{inhomogeneous} formulas.

\begin{corollary}
\label{cor:inhom}
Let $d = O(\log n)$. Then a homogeneous polynomial $P$ defined on $n$ variables with degree $d \geq 2$ has an algebraic formula of size $\poly(n)$ if and only if it has an algebraic formula of depth $O(\log d)$ and size $\poly(n).$
\end{corollary}

In particular, this means that to prove superpolynomial lower bounds against general algebraic formulas in the low-degree setting, it suffices to prove such lower bounds against homogeneous algebraic formulas of depth $O(\log d).$ As far as we know, nothing below the trivial $O(d)$ bound was known before for such an implication. This brings us much closer to the regime of depths for which we have lower bounds~\cite{LST-FOCS}.

The statements are even starker in the non-commutative setting, where it is a long-standing problem to prove separations between \emph{Algebraic Branching Programs} (ABPs) and formulas. In recent work~\cite{LST-STOC}, it was shown how to prove such a result for depths that are $o(\sqrt{\log d})$. The results of this paper show that it suffices to prove such a result for depth $O(\log d).$\footnote{We note that Raz's result, though only stated for the commutative setting, works just as well in the non-commutative case.}

Finally, we also show that our results cannot be improved asymptotically in terms of depth, unless we use techniques that exploit cancellations in some way.

\begin{theorem}[Lower Bound]
\label{thm:lowerbound-intro}
Let $n$ and $d = d(n)$ be growing parameters such that $d(n)\leq \sqrt{n}$. Then there is a monotone algebraic formula $F$ of size at most $n$ and depth $O(\log d)$ computing a polynomial $P\in \F[x_1,\ldots,x_n]$ of degree at most $d$ such that any monotone formula $F'$ of depth $o(\log d)$ computing $P$ must have size $n^{\omega(1)}.$
\end{theorem}

It should be noted that a well-known result of Gupta, Kamath, Kayal and Saptharishi~\cite{GKKSdepth3} shows how to exploit cancellations to obtain better depth-reduction results. However, in general this does not reduce the depth of a given formula by more than a constant factor without incurring a significant blow-up in size.\footnote{More precisely, the result of~\cite{GKKSdepth3} shows how to convert a low-degree \emph{homogeneous} depth-$4$ formula to an \emph{inhomogeneous} depth-$3$ formula efficiently. In general, this can be used to reduce the depth of a small-depth formula by a multiplicative factor of $2$.} As a result, we believe that the above result is a strong indication that our depth reduction result is tight up to a constant factor in the depth.
\section{Preliminaries}
\label{sec:prelims}

\paragraph{Basic notation.} Throughout, unless otherwise specified, we work with polynomials over a field $\F$. We will work with the multivariate ring of polynomials $\F[x_1,\ldots,x_n]$ or its \emph{non-commutative} analog $\F\langle  x_1,\ldots,x_n \rangle.$ 

\subsection{Algebraic formulas} 
\label{sec:formula-prelims}

We start with some brief definitions and results related to algebraic formulas. For much more about this model, see the standard references~\cite{SY,Ramprasadsurvey}.

\paragraph{The model.} An algebraic formula over the multivariate polynomial ring $\F[x_1,\ldots,x_n]$ is a rooted, directed tree with edges directed towards the root. Leaves are labelled by variables $x_1,\ldots,x_n$ or by the constant \(1\) and edges by non-zero field constants. Internal nodes (i.e.,~gates) by $+$ and $\times$ and compute linear combinations (based on the edge weights) or products of their children. We will assume, with loss of generality, that if a node \(\alpha\) has for child a leaf labelled by \(1\), then \(\alpha\)  is a \(+\)-gate and that if a \(+\)-gate \(\alpha\) has only children labelled by \(1\), then \(\alpha\) is the output of the formula.\footnote{This ensures that a formula can compute polynomials with a constant term but forbids using many arithmetic operations just to compute constants.}
A \emph{non-commutative} algebraic formula  over the multivariate polynomial ring $\F\langle x_1,\ldots,x_n\rangle$ is defined similarly, with the additional assumption that the children of any $\times$-gate are linearly ordered, and the corresponding product is computed in this order.

Unless explicitly stated, the algebraic formulas we consider have \emph{unbounded} fan-in (i.e.,~a gate can have any number of inputs). The \emph{size} of $F$ will denote the number of leaves,\footnote{\label{ft:fanin}This is within a constant factor of the number of gates, as long as each gate has fan-in at least $2$ each (which is without loss of generality).} the \emph{depth} of $F$ the longest leaf-to-root path. The \emph{product-depth} and the \emph{sum-depth} of $F$ are defined to be the maximum number of product gates and sum gates encountered on a leaf-to-root path, respectively.

A \emph{parse tree} of a formula $F$ is a subformula of $F$ which corresponds to the way a monomial is built in the evaluation of $F$.
Parse trees of $F$ can be defined inductively as follows:
\begin{itemize}
\item If $F$ has a top $+$-gate,
a parse tree of $F$ is obtained by taking a parse tree of one of its children together with the corresponding edge to the root of $F$;
\item If $F$ has a top $\times$-gate,
a parse tree of $F$ is obtained by a taking a parse tree of each of its children, together with the incoming edges of the root of $F$;
\item The only parse tree of a leaf is itself.
\end{itemize}
The polynomial computed by a parse tree is a single monomial, which
is equal to the product of the variables labelling its leaves, multiplied
by the product of the scalars labelling its edges.
The polynomial computed by a formula $F$ is easily seen to be the sum of the monomials computed by all its parse trees.

A formula $F$ is called \emph{monotone} if any monomial computed by a parse tree of $F$ has a non-zero coefficient
in the polynomial computed by $F$.

We now recall some well-known results from the literature regarding algebraic formulas. It should be noted that these results (specifically Theorems~\ref{thm:standardDR}, \ref{thm:Raz} and \ref{thm:BCE} later on) are usually proved in the general setting of commutative formulas. However, it is easy to see that the proofs of these results carry over to the monotone, non-commutative setting without significant change.

\paragraph{Depth-reduction.} Classical results~\cite{BKM,Brent} show that any algebraic formula of small size can be simulated by one of small depth and not much larger size. Formally,

\begin{theorem}
\label{thm:standardDR}
 Let $F$ be a (non-commutative or commutative) algebraic formula of size $s$ computing a polynomial $P$. Then there is an algebraic formula $F'$ of size at most $\poly(s)$ and depth $\Delta = O(\log s)$ computing $P$. We may also assume that each gate in $F'$ has fan-in $2$. Furthermore, $F$ is homogeneous and/or monotone, then so is $F'$.
\end{theorem}

\paragraph{Homogeneity.} Each gate in an algebraic formula has a \emph{syntactic degree} defined in a natural way. Leaves labelled by the constant~$1$ have syntactic degree~\(0\),  leaves labelled with a variable have syntactic degree~$1$, $\times$-gates have a syntactic degree that is the sum of the syntactic degrees of their children, and $+$-gates have a syntactic degree that is equal to the largest of the syntactic degrees of their children.
 The syntactic degree of a formula is defined as the syntactic degree of its output. Notice that in a formula  the syntactic degree of any gate is bounded by the syntactic degree of the formula. 

We will further assume that no gate computes the zero polynomial.

A formula is \emph{homogeneous} if each gate in the formula computes a homogeneous polynomial. Equivalently, in terms of syntactic degrees, this means that all the children of a sum gate have the same syntactic degree.

Raz~\cite{Raztensor} showed how to convert a possibly inhomogeneous formula $F$ to a homogeneous formula with a relatively small blow-up in size.

\begin{theorem}
\label{thm:Raz}
 Let $F$ be a (non-commutative or commutative) algebraic formula of size $s$ and product-depth $\Delta$ computing a polynomial $P$ such that all gates in $F$ have fan-in $2$. Then there is a \emph{homogeneous} algebraic formula $F'$ of size at most $O\left(s\cdot \binom{\Delta + d+1}{d}\right)$ and product-depth $\Delta$ computing $P$.  In particular, if $\Delta = O(\log s)$ and $d = O(\log s)$, then the formula $F'$ has size $\poly(s).$
\end{theorem}

\section{Main result}
\label{sec:mainlbd}

\paragraph{Proof Overview.} While the proof of the main result is fairly short and (in our opinion) clean, we add some remarks here to clarify why previous depth-reduction proofs are not applicable in our setting.

The first attempt in proving Theorem~\ref{thm:main-intro} would be to try to use the proof strategy behind Theorem~\ref{thm:standardDR}. Here, we start with a formula $F$ of size $s$ and find a subformula $G$ of size roughly $s/2$ rooted at some gate $\alpha$ of $F$. It is not hard to show that the polynomial computed by $F$ can be written as
\[
F = G\times H_1 + H_2
\]
where $H_1$ and $H_2$ are also computed by formulas of size $s/2$.\footnote{Here, for simplicity, we are assuming that $F$ is a commutative formula. In the non-commutative setting, we would instead get $F = H_1' G H_1'' + H_2.$} We then apply induction to these three subformulas to get the result. Unfortunately, this strategy does not use the degree of the formula at all, and therefore only yields a formula of depth $O(\log s).$

In the homogeneous setting, it is sometimes more natural to do induction on the degree of the underlying formula, in which case $G$ and $H_1$ (in the decomposition above) would be subformulas of degree roughly $d/2$. We get the following recursion on the worst-case size of the depth-reduced version of $F$, which we denote by $T(s,d)$
\[
T(s,d) \leq T(s_1,d/2) + T(s-s_1,d/2) + T(s-s_1,d)
\]
where $s_1$ denotes the size of $G$. Unfortunately, in this case, the formulas $H_1$ and $H_2$ \emph{could} have size nearly $s$, resulting in considerable size blow-up. Indeed, when $s_1$ is much smaller than $s$ (say $s_1 = s^{o(1)}$), the above recursion only yields $T(s,d) = s^{O(\log d)}$, which is a superpolynomial size blow-up.

It may be possible to interleave recursions with respect to size and depth, but we were unable to make this work. 

Another possible strategy could be to follow the work of~\cite{VSBR} which produces \emph{circuits} of the required depth. Unfortunately, the proof of~\cite{VSBR} is a memoization procedure, which seems to yield circuits even when applied to a formula $F$. Turning the resulting circuit of depth $O(\log d)$ into a formula seems to increase the  size to $s^{\Omega(\log d)}$.

The approach we take is somewhat more global than the recursive strategies outlined above. Our first motivating example is seemingly the worst-case example for depth-reduction: a \emph{comb} of depth greater than $d$ with alternating sums and products. More formally, a comb computes the following polynomial (up to identifying variables).
\[
C(x) = x_1 + (x_2 \times (x_3 + ( x_4 \times \cdots )))
\]
Note that the above yields a formula of depth greater than $d$ where $d$ is the degree of the underlying polynomial. However, we observe that such a comb actually computes a polynomial with only a few monomials, and hence can be written trivially as a depth-$2$ $\sum\prod$-formula without much of a size blow-up. 

Building on this observation, the overall strategy is to decompose the formula into a top part $G$, which is a (generalized) comb, whose leaves are subformulas of $F$ to which we will apply the same procedure recursively. We then write $G$ as a $\sum\prod$-formula, and replace its leaves with the depth-reduced versions of the subformulas. This gives the depth-reduced version of $F$.

The correct definition of $G$ is crucial, and somewhat subtle (at least to us), but with the proper definitions in place, the proof goes through without much trouble.

\subsection{Proof of Theorem~\ref{thm:main-intro}}

In this section we prove our main result (Theorem~\ref{thm:main-intro}). 
We start by showing a simple depth-reduction result for the case of \emph{skew formulas}. A formula is said to be skew if every multiplication gate in it has at most $1$ non-trivial child (i.e., a non-leaf node).

We will say that a leaf in a skew-formula is a $+$-leaf if the parent of that leaf is a $+$ gate and a $\times$-leaf otherwise. 
We show that any skew formula can be converted efficiently into a depth-$2$ formula, i.e., a $\sum\prod$-formula.

\begin{lemma}
	\label{lem:skew}
	Let $G$ be a skew formula with sum-depth $\delta$, wherein all the gates have fan-in $2$ and the leaves are labelled by distinct variables. Then the polynomial computed by $G$ is a multilinear polynomial with at most $2^\delta$ monomials. Moreover any variable labelling in \(G\) 
	\begin{itemize}
		\item a $+$-leaf,
		\item or a \(\times\)-leaf whose sibling is a leaf
	\end{itemize}
	appears in exactly one monomial. We will call them the {\em non-duplicable} variables. 
\end{lemma}
\begin{proof}
	We prove this by induction on the depth of the formula. The base case is when the depth is $0$ or \(1\).
	In both cases, the statement trivially holds. 
	
	For the induction case, assume that the depth is at least \(2\). 
	Suppose $G = G_1 + G_2$. The sum-depths of $G_1$ and $G_2$ are at most $\delta-1$. Let $f_1,f_2$ be the multilinear polynomials computed by $G_1, G_2$, respectively. We know that the leaves of $G$ are labelled with distinct variables. Hence, $G_1, G_2$ have the same property and the variable sets labelling the leaves of $G_1$ and $G_2$ are disjoint. 
	The depths of $G_1$ and $G_2$ are strictly smaller than the depth of $G$. By the induction hypothesis we have that $f_1$ and $f_2$ have at most $2^{\delta-1}$ monomials.  Moreover, any non-duplicable variable in $G_1$ appears in at most one monomial in $f_1$. Similarly, any non-duplicable variable in $G_2$ appears in at most one monomial in $f_2$. Hence, the polynomial computed by $G$, i.e., $f_1+f_2$,  has at most $2^\delta$ monomials and each non-duplicable variable appears in at most one monomial in it. 
	
	Suppose the top gate of $G$ is a $\times$ gate. As $G$ is a skew formula it is either  $x \times G_1$ or $G_1 \times x$, where $x$ is a variable. In particular since the depth of \(G\) is at least two, the variable \(x\) is duplicable. Let $f_1$ be the multilinear polynomial computed by $G_1$. By our assumption, the variable $x$ does not appear in $f_1$. The depth of $G_1$ is strictly smaller than $G$.  By the induction hypothesis we have that $f_1$ has at most $2^\delta$ monomials and any non-duplicable variable in $G_1$ appears in at most one monomial. As $x$ can distribute over the monomials of $f_1$, we have that the polynomial computed by $G$ is multilinear with at most $2^\delta$ monomials and any non-duplicable variable appears in at most one monomial.  
	
	(Note that this can also be seen using parse trees, since the polynomial computed is the sum of the monomials computed by the parse trees. The multilinearity is obvious. Note that parse trees do not really ``branch" at multiplication gates because of the skewness and are therefore combs.  To build a parse tree starting from the root we will have two choices for each addition gate we encounter on the path, and there are at most $\delta$, so we get at most $2^\delta$ parse trees. The only parse tree containing a given non-duplicable variable is defined by the path from the root to this leaf.)
\end{proof}

If a formula is homogeneous, it implies that for any gate \(\alpha\), the degree of the polynomial computed by \(\alpha\) coincides with the syntactic degree \(d_{\alpha}\) of \(\alpha\). Based on this remark, below we prove a stronger statement than Theorem~\ref{thm:main-intro}. Specifically, Theorem~\ref{thm:main-intro} is stated for homogeneous formulas. But here, we show a depth-reduction for formulas for which the syntactic degree is small. 

\begin{theorem}[Refinement of  Theorem~\ref{thm:main-intro}]
	\label{thm:main-syntactic-degree}
	Let $F$ be an algebraic formula of size $s$ and of syntactic degree \(d_F \geq 2\). Then $P$ is also computed by a formula $F'$ of size $\poly(s)$ and depth $O(\log d_F)$. Moreover, if $F$ is homogeneous, monotone and/or non-commutative, then so is $F'$.
\end{theorem}

\begin{proof}
	Let us start with a formula obtained from $F$ after applying Theorem~\ref{thm:standardDR}. That is, we will assume that we have a formula of size $\poly(s)$ such each gate in it has fan-in~$2$, sum-depth and product-depth are bounded by $O(\log s)$. For notational simplicity, from now on, $F$ will refer to this new formula.

	Let \(\delta\) be a positive integer. For a formula $G$ of syntactic degree $d_G \geq 1$ and sum-depth $\Delta(G)$ we define a potential function $\phi_\delta(G)$ as follows.
	\begin{equation*}
		\left\{\begin{aligned}
			\phi_{\delta,1} (G) & = \lceil \log(d_G) \rceil \\
			\phi_{\delta,2} (G) & =  \lceil \Delta(G) / \delta \rceil
		\end{aligned}\right.
	\end{equation*}
	and let $$\phi_{\delta}(G)= \phi_{\delta,1} (G) + \phi_{\delta,2} (G).$$
	
	We will show that the potential function bounds the depth of the depth-reduced formula that we will construct. We will also use it to bound the size of the resulting formula. Specifically, we prove the following lemma. 
	
	\begin{lemma}
		\label{lem:main-lemma}
		Let \(\delta\) be a positive integer.
		Any formula $F$ of fan-in \(2\), syntactic degree $d \geq 1$, sum-depth $\Delta$, and size \(s\) can be parallelized into a formula $F'$ (of arbitrary fan-in) of product-depth at most \(\phi_{\delta}(F)\)  and size at most
		\(s\cdot 2^{\delta \log(d)}\).
		Further, if $F$ is homogeneous, monotone and/or non-commutative, so is $F'$.
	\end{lemma}
	
	Since we ensured that the sum-depth is bounded by \(O(\log s)\), taking
	\(\delta=\left\lceil \frac{\log s}{\log d} \right\rceil \)
	and applying Lemma~\ref{lem:main-lemma}, we get that the final formula $F'$ has size at most \(\poly(s)\) and product-depth at most 
	\[
	\phi_{\delta}(F) = O\left( \lceil \log d \rceil+ \left\lceil \log(s) \frac{\log d}{\log s} \right\rceil \right) =  O(\log d).
	\]
	By collapsing sum gates that feed into other sum gates, we see that the depth of $F'$ can be assumed to be at most twice its product-depth, which is $O(\log d)$. This thus finishes the proof of the theorem. \end{proof}

We now prove Lemma~\ref{lem:main-lemma}. 
\begin{proof}[Proof of Lemma~\ref{lem:main-lemma}]
	For any gate $\alpha$ in $F$, let $F_\alpha$ denote the subformula rooted at $\alpha$ and \(d_{\alpha}\) be its syntactic degree. We do the proof by induction on \(\phi_{\delta}(F)\).
	
	The base case $\phi_\delta(F) = 0$ trivially holds. 
	Consider the following set of gates of~$F$: 
	\[ \mathcal{A} = \left\{\alpha \mid \alpha \text{ is a not a leaf labelled~$1$  and }\phi_{\delta}(F_\alpha)  < \phi_{\delta}(F) = \phi_{\delta}(F_{\text{parent}(\alpha)})\right\}.\] 
	For any gate \(\alpha\) in \(\mathcal{A}\), the induction hypothesis tells us that we can construct a formula $F_\alpha'$ of product-depth at most \(\phi_{\delta}(F)-1\)  and size at most \(s_{\alpha}\cdot 2^{\delta \log(d_\alpha)}\) computing the same polynomial as $F_\alpha.$
	
	Let us consider the formula \(G\) obtained by replacing these gates from $\mathcal{A}$ in \(F\)  by leaves (labelled with distinct variables). Notice that for a product-gate \(\beta\) in $F$, at most one of its children has syntactic degree larger than \(d_{\beta}/2\), where $d_\beta$ is the syntactic degree of $\beta$. Consequently, \(G\) is a skew formula.  The other child of $\beta$ is a $\times$-leaf in $G$. Moreover, \(G\) has sum-depth at most \(\delta\) (since \(\phi_{\delta,2}\) strictly decreases for gates below). 
	
	Hence, we can use Lemma~\ref{lem:skew} to simplify $G$: we get that the polynomial computed by $G$ is a multilinear polynomial in its leaves and has $2^\delta$ monomials. We can then  write $G$ as a \(\sum\prod\)-formula $G'$    such that each duplicable gate appears in  at most $2^\delta$ monomials and each non-duplicable gate in at most~$1$.
	
	The new formula $F'$ for $F$ is obtained from $G'$ by replacing each variable leaf, which corresponds to some gate $\alpha$ in $F$, by its depth-reduced version $F'_\alpha$ constructed using the induction hypothesis above.

	The product-depth of $F'$ is bounded by the product-depth of the gates $\alpha \in \mathcal{A}$ plus the product-depth of $G$, which is equal to $1$ after rewriting it as a $\sum\prod$-formula. That is, the product-depth is at most  \((\phi_{\delta}(F)-1)+1 = \phi_{\delta}(F)\). By construction if \(G\) does not contain  leaves labelled by~$1$, then it is also the case for \(G'\), otherwise, \(G'\) still has at most one such leaf.
	The size of the resulting formula is bounded by
	\[ 	\sum_{\alpha\ \text{non-duplicable }} \left(s_\alpha \cdot 2^{\delta \log(d_\alpha)} \right)  
	+
	\left( 2^{\delta} \cdot \sum_{\alpha\ \text{duplicable}} \left(s_\alpha \cdot 2^{\delta \log(d_\alpha)} \right) \right) 
	+ \mathbbm{1}_{\text{G' has a constant leaf}} . \]
	Notice that if $\alpha$ is duplicable, it must be a \(\times\)-leaf with its sibling $\beta$ not a leaf. This
	means that $d_\alpha \leq d_\beta \leq d$ since the syntactic degree is maximal at the root. Hence $d_\alpha \leq d/2$
	so the contribution of duplicable gates is bounded by
	$$\sum_{\alpha\ \text{duplicable }} 2^\delta \left(s_\alpha \cdot 2^{\delta \left(\log(d)-1\right)} \right)
	\leq \sum_{\alpha\ \text{duplicable }} \left(s_\alpha \cdot 2^{\delta \log(d)} \right).$$
	The
	contribution of non-duplicable gates is bounded by
	$$\sum_{\alpha\ \text{non-duplicable }} s_\alpha \cdot 2^{\delta \log d}$$
	since any gate in $F$ has syntactic degree at most $d$.
	By the choice of $\mathcal{A}$, the subformulas $F_\alpha$ are disjoint
	so $\sum_{\alpha \in \mathcal{A}} s_\alpha \leq s$, with strict inequality if $G'$ has a constant leaf.
	Hence
	the size of $F'$ is bounded by $s\cdot 2^{\delta \log(d)}$.
	
	Finally, it is straightforward to verify that the construction preserves homogeneity, monotonicity and/or non-commutativity.
\end{proof}

\begin{remark}
	\label{rmk:degree-not-increase}
	It is easy to note that our depth reduction procedure does not increase the syntactic degree of the formula. 
\end{remark} 

We also observe that putting Theorem~\ref{thm:main-intro} together with Theorems~\ref{thm:standardDR} and~\ref{thm:Raz} immediately implies Corollary~\ref{cor:inhom}.

\begin{proof}[Proof of Corollary~\ref{cor:inhom}]
	Given a (possibly inhomogeneous) formula $F$ of size $s = \poly(n)$ computing a polynomial of degree $d = O(\log n),$ we first apply standard depth-reduction (Theorem~\ref{thm:standardDR}) to get an equivalent formula $F_1$ of size $s_1 = \poly(n)$ and depth $\Delta_1 = O(\log n)$ where each gate has fan-in $2$. Applying Raz's homogenization theorem (Theorem~\ref{thm:Raz}) to $F_1$ yields an equivalent \emph{homogeneous} formula $F_2$ of size $s_2 = \poly(n)$ and depth $\Delta_2 = O(\log n)$. We can now apply Theorem~\ref{thm:main-intro} to $F_2$ to get an equivalent formula $F'$ of size $s' = \poly(n)$ and depth $\Delta' = O(\log d).$
\end{proof}

\subsection{Reducing the size blow-up}
\label{sec:nearlinear}

We note that the above strategy can be easily adapted to yield a small depth formula of size $s'$ that is nearly linear in the size $s$ of the original formula, at the expense of increasing the depth by a large constant. 

The proof is nearly identical to the proof of Theorem~\ref{thm:main-intro} above. The new ingredient is a near-linear depth-reduction in the setting where there is no bound assumed on the degree of the above formula. More precisely, Bshouty, Cleve and Eberly~\cite{BCE} (see also the work of Bonet and Buss~\cite{BonetBuss}) showed the following (we sketch Bonet and Buss' proof in Appendix~\ref{sec:BonetBuss} for completeness).

\begin{theorem}[Depth-reduction with near-linear size]
	\label{thm:BCE}
	The following holds for any $\varepsilon > 0.$ Let $F$ be a (non-commutative or commutative) algebraic formula of size $s$ computing a polynomial $P$. Then there is an algebraic formula $F'$ of size at most $s^{1+\varepsilon}$ and depth $\Delta = 2^{O(1/\varepsilon)}\cdot \log s$ computing $P$. We may also assume that each gate in $F'$ has fan-in $2$. Furthermore, if $F$ is homogeneous and/or monotone, then so is $F'$.
\end{theorem}

Using the above result, we can prove the following improved version of our depth-reduction.  

\begin{theorem}
	\label{thm:strongermain}
	Assume that $F$ is a (commutative or non-commutative) formula of size $s$ and syntactic degree $d \geq 1$ computing a polynomial $P$. Then $P$ is also computed by a formula $F''$ of size at most $s^{1+\varepsilon}$ and depth $\Delta = 2^{O(1/\varepsilon)}\cdot \log d.$ Furthermore, if $F$ is a homogeneous and/or monotone formula, then so is $F''$.
\end{theorem}	

\begin{proof}
	If \(d^{4/\varepsilon} \geq s\), then the depth-reduction of Theorem~\ref{thm:BCE} already gives a satisfactory solution. Indeed the size is bounded by \(s^{1+\varepsilon}\) and the depth is bounded by
	\[
	2^{O(1/\varepsilon)} \log s \leq \frac{4}{\varepsilon} 2^{O(1/\varepsilon)}\log d \leq 2^{O(1/\varepsilon)} \log d.
	\]
	
	So we assume that \(s > d^{4/\varepsilon}\). By first applying Theorem~\ref{thm:BCE} (with $\varepsilon/2$ instead of $\varepsilon$), we obtain a formula $F'$ of size at most $s^{1+\varepsilon/2}$, depth $\Delta' = 2^{O(1/\varepsilon)}\cdot \log s$, and fan-in $2$ computing $P$.
	
	Now, we apply Lemma~\ref{lem:main-lemma}, while setting \(\delta = \lfloor {(\varepsilon \log s)}/{(2 \log d)} \rfloor\). 	 
	Notice that since \(\varepsilon \log s > 4 \log d\), it ensures that \(\delta  > (\varepsilon \log s)/(4 \log d) > 1\).

	Then the strategy produces a formula $F''$ of product-depth at most
	\[
	\phi_\delta(F')\leq \phi_{\delta,1}(F') + \phi_{\delta,2}(F') <
	\lceil\log d\rceil + 2^{O(1/\varepsilon)}\log s \frac{4 \log d}{\varepsilon \log s} = 2^{O(1/\varepsilon)}\log d
	\]
	and size at most 
	\[
	s^{1+\varepsilon/2}\cdot 2^{\delta \log d} \leq s^{1+\varepsilon/2} s^{\varepsilon/2}.
	\]
	This proves the theorem.	
\end{proof}

\subsection{Reducing the product fan-ins to 2}
\label{sec:fanin2}

It is natural to ask if Theorem~\ref{thm:main-intro} can be proved while ensuring that the fan-in of each gate is bounded by~$2$, as in Theorem~\ref{thm:main-intro} and Theorem~\ref{thm:BCE}. This is not possible, as a formula of fan-in $2$ and depth $O(\log d)$ can only compute polynomials on at most $\poly(d)$ variables, while formulas of size $s$ may have up to $s$ variables. However, this does not rule out reducing the fan-in of the product gates to $2$. Indeed, the \emph{circuit} depth-reduction of~\cite{VSBR} does exactly this. We show now that this can also be done for algebraic formulas with bounded syntactic degree.

\begin{theorem}
	\label{thm:prodfanin}
	Let $F$ be a (commutative or non-commutative) algebraic formula $F$ of size $s$, depth $\Delta$, and syntactic degree \(d \geq 1\) computing a polynomial $P$. Then $P$ can also be computed by a formula $F'$ of size $s$ and depth $\Delta' = O(\Delta + \log d)$ where each product gate of $F'$ has fan-in $2$. Furthermore, if $F$ is a homogeneous and/or monotone formula, then so is $F'$.
\end{theorem}

Plugging this in Theorem~\ref{thm:main-syntactic-degree} 
gives a depth-reduction to formulas of depth $O(\log d)$ and size $\poly(s)$ such that all product gates have fan-in at most $2$. A similar result can be obtained with a smaller blow-up in size by combining this statement with Theorem~\ref{thm:strongermain}.

\begin{proof}
	It suffices to prove a weaker version of the above theorem where each product gate has fan-in at most $3$. We can then replace each of the products of fan-in $3$ by a tree of product gates of fan-in $2$ and size $3$. This has the effect of increasing the depth at most by a factor of $2$, which does not affect the overall result. 
	
	So we will prove this slightly weaker version. In this setting, we will aim for a depth $\Delta' = \Delta + \log d.$
	
	This is done by induction on the depth $\Delta$ of the formula. The case of $\Delta = 0$ is trivial. Let $F$ be a formula of depth $\Delta > 0$ and syntactic degree $d$.
	
	Assume that the output gate of $F$ is a sum gate, and $F_1,\ldots, F_t$ are the subformulas of $F$ of depth $\Delta -1$. By definition, each $F_i$ has syntactic degree at most $d$. Applying the induction hypothesis to each of the $F_i$ yields a formula $F_i'$ with product gates of fan-in at most $3$. The formula $F'$ can then be defined as the sum of these formulas. 
	
	Now we come to the main case, which is when the output gate of $F$ is a product gate. Assume that $F_1,\ldots, F_t$ are the subformulas of $F$ of depth $\Delta -1$ in the order\footnote{The order is important in the non-commutative setting.} that they appear in $F$. Let $d_i$ denote the syntactic degree of $F_i$. Define 
	\[
	m = \min\{j \ | \sum_{i= 1}^j d_i \geq d/2\}.
	\]
	Let $F_\ell$ be the formula obtained from $F$ by keeping only the subformulas $F_1,\ldots,F_{m-1}$, and $F_r$ be the formula obtained by keeping only $F_{m+1},\ldots,F_t.$ We use the induction hypothesis on $F_\ell, F_m,$ and $F_r$ to get formulas $F_\ell', F_m'$ and $F_r'$. Finally, we set 
	\[
	F' = F_\ell' \times F_m' \times F_r'.
	\]
	The size\footnote{Recall that the size of a formula is the number of its leaves.} of $F'$ is the sum of the sizes of $F_\ell', F_m'$ and $F_r'$, which is at most $s$ by the induction hypothesis. Let $\Delta_\ell', \Delta_m'$ and $\Delta_r'$ denote the depths of $F_\ell', F_m'$ and $F_r'$ respectively. The depth of $F'$ is 
	\[
	1+\max\{ \Delta_\ell', \Delta_m', \Delta_r'\} \leq 1+ \max\{\Delta + \log(d/2), \Delta-1 + \log d, \Delta + \log(d/2)\} = \Delta + \log d
	\]
	where the second inequality uses the induction hypothesis, and the fact that $F_{\ell}$ and $F_r$ have syntactic degree at most $d/2$ and $F_m$ has depth at most $\Delta-1.$
\end{proof}

\section{Tightness}

\label{sec:tightness}

Given integers $k \geq 1$ and $r \geq 2$, we will define a polynomial $H^{(k,r)}$.
Intuitively, we want to define this polynomial as a standard universal polynomial for formulas. It is composed of \(k\)-nested inner products, each one of size \(r\).
In the following we will drop the superscript in $H^{(k,r)}$ and write simply $H$ instead.

The polynomial $H$ will be defined over the set of
$(2r)^k$  variables
$$\{x_{\sigma,\tau}\ |\ \sigma\in [2]^k, \tau \in [r]^{k}\}.$$
Let us define recursively polynomials $H_{u,v}$ for all $(u,v) \in [2]^{\leq k} \times [r]^{\leq k}$ such that $|u|=|v|$:
\begin{align*}
	H_{u,v} = & \ x_{u,v} 	&& \text{when }|u|=|v|=k\\
	H_{u,v} = & \sum_{a=1}^{r} H_{u 1,v a} H_{u 2, v a} && \text{otherwise}.
\end{align*}
The polynomial $H$ is defined as the polynomial \(H_{\varepsilon,\varepsilon}\).
Note that $H$ is a polynomial of degree $d=2^k$ and has $r^{d-1}$ monomials.

From its definition, $H$ is computed by a monotone formula $M$
of size $(2r)^k$
and depth $2k$,
with a $+$-gate at the top, alternating layers of $+$-gates and $\times$-gates, with $+$-gates of fan-in $r$
and $\times$-gates of fan-in $2$, and leaves labelled with distinct variables.

For words $u$ and $v$ over the same alphabet, we write $u \prefix v$ if $u$ is a prefix of $v$.
There is a natural one-to-one correspondance between prefixes of words of $([r]\times[2])^k$
and nodes of $M$, which is the following.
Let $\sigma = \sigma_1 \ldots \sigma_k \in [2]^k$, and $\tau = \tau_1 \ldots \tau_{k} \in [r]^{k}$.
The word $\tau_1 \sigma_1 \ldots \tau_k \sigma_k$ corresponds to a path from the root of $M$
to the leaf labelled $x_{\sigma,\tau}$, while proper prefixes of $\tau_1 \sigma_1 \ldots \tau_k \sigma_k$
correspond to internal gates in $M$ along this path.
For $\ell < k$, $u=u_1\ldots u_\ell \in [2]^\ell$ and $v=v_1 \ldots v_\ell \in [r]^\ell$,
the node which corresponds to the word $v_1 u_1 \ldots v_\ell u_\ell$ is the $+$-gate of $M$ computing $H_{u,v}$.

The polynomial $H$ is easily seen to be set-multilinear with respect to
the sets of variables $\{ X_\sigma\ |\ \sigma \in [2]^k\}$
where $X_\sigma = \{x_{\sigma,\tau}\ |\ \tau \in [r]^{k} \}$. This means that each monomial has exactly one variable from each set $X_\sigma$.

\begin{remark}
	\label{rem:isomorphic-subpolynomials}
	For $|u|=|v|=\ell \leq k$, $u \in [2]^\ell$ and $v \in [r]^\ell$,
	the polynomial $H^{(k,r)}_{u,v}$
	is defined over the set of variables
	$$X_{u,v} = \{ x_{\sigma,\tau}\ |\ \sigma \in [2]^k,\ \tau \in [r]^{k}, u \prefix \sigma,\ v\prefix \tau\}$$
	and is the polynomial $H^{(k-\ell,r)}$ (upto renaming of the variables).
	In particular, its degree is $d'=2^{k-\ell}$,
	it has $r^{d'-1}$ monomials 
	and it is set-multilinear with
	respect to $\{ X_\sigma\ |\ \sigma \in [2]^k,\ u\prefix \sigma\}$.
\end{remark}

Before proving hardness of $H$ for small-depth monotone formulas, we need to show that the gates of a monotone formula computing~\(H\) cannot compute too many monomials. This is proved in Lemma~\ref{lem:number-of-monomials} below.

\begin{proposition}
	\label{prop:product-of-variables-in-H}
	Consider two variables $x_{\sigma,\tau}$ and $x_{\sigma',\tau'}$.
	If $x_{\sigma,\tau} x_{\sigma',\tau'}$ appears in a monomial of~$H$,
	and if $\sigma$ and $\sigma'$ have a common prefix of length $\ell < k$,
	then $\tau$ and $\tau'$ have a common prefix of length $\ell+1$.
\end{proposition}

\begin{proof}
	Observe that
	$x_{\sigma_1,\tau_1} \ldots x_{\sigma_p,\tau_p}$ is a monomial of $H$ if and only if
	these variables
	form the leaves of a parse tree of $M$.
	As observed above, the root-to-leaf path leading to the variable
	$x_{\sigma_i,\tau_i}$ in $M$
	is obtained by taking in turn the first letter of  $\tau_i$, the first letter of $\sigma_i$, the second letter of $\tau_i$, etc.
	
	If the product $x_{\sigma,\tau} x_{\sigma',\tau'}$ appears in a monomial of $H$, it must
	be possible to complete
	the union of the two paths, from the root to $x_{\sigma,\tau}$ and from the root to $x_{\sigma',\tau'}$,
	into a parse tree of the formula $M$.
	
	If the longest common prefix of $\tau$ and $\tau'$ were of length at most $\ell$,
	then the lowest common ancestor of $x_{\sigma,\tau}$ and $x_{\sigma',\tau'}$ in $M$ would be a $+$-gate,
	which is not possible in a parse tree where each $+$-gate has a single child.
\end{proof}

\begin{lemma}\label{lem:number-of-monomials}
	If \(F\) is a monotone formula which computes the polynomial \(H\) and if \(\alpha\) is a gate of \(F\) of degree \(d_{\alpha}\), then the number of monomials of the polynomial computed at gate \(\alpha\) is at most \(r^{d_{\alpha}-1}\).
\end{lemma}

\begin{proof}
	Since \(H\) has \(r^{d-1}\) monomials, the result is true when \(d_{\alpha} = d\) by monotonicity. Assume now that \(d_{\alpha} < d\). 
	
	Let $$I = \{ \sigma \in [2]^k\ |\ \text{some variable of $X_\sigma$ appears in $\alpha$}\}.$$
	For $u \in [2]^{\leqslant k}$, let $I_u = \{ \sigma \in [2]^k\ |\ u \prefix \sigma\}$.
	Let $\{u_1, \ldots, u_p\}$ be the set of words $w$ of minimal length such that $I_w \subseteq I$.
	Then $I$ is the disjoint union $\bigcup_{\ell \in [p]} I_{u_\ell}$.
	Since $d_\alpha<d$, $I \neq [2]^k$ so no $u_\ell$ is the empty word.
	
	Consider some $\ell \in [p]$.
	Let $\bar{u}_\ell$ be obtained by switching the last letter of $u_\ell$.
	By minimality of the length of $u_\ell$, we must have $I_{\bar{u}_\ell} \nsubseteq I$. 
	Let $\sigma$ be a word in   $I_{\bar{u}_\ell} \setminus I$. 
	Since $F$ is monotone and computes a set-multilinear polynomial, it is a set-multilinear formula
	and therefore the polynomial computed at $\alpha$
	must be multiplied by some variable $x_{\sigma,\tau}$.
	Let $v_\ell$ be the prefix of length $|u_\ell|$ of $\tau$.
	Consider any variable $x_{\sigma',\tau'}$ appearing in $\alpha$ such that $u_\ell \prefix \sigma'$.
	Since the product $x_{\sigma,\tau} x_{\sigma',\tau'}$ must appear in a monomial of $H$ by monotonicity,
	it must be that $v_\ell \prefix \tau'$
	by Proposition~\ref{prop:product-of-variables-in-H}.

	Any monomial $m$ in the polynomial computed in gate $\alpha$ can be written in a unique way $m=m_1 \cdots m_p$ with $m_\ell$
	set-multilinear with respect to $\{X_{\sigma}\ |\ u_\ell \prefix \sigma\}$.
	By the above, $m_\ell$ is a monomial over the variables $X_{u_\ell,v_\ell}$ of degree $|I_{u_\ell}|$.
	By monotonicity, it should be possible to complete leaves of $M$ labelled with variables from $m_\ell$
	into a parse tree of $M$ appearing in $H$, which proves that 
	$m_\ell$ is a monomial of $H_{u_\ell,v_\ell}$.
	There are at most $r^{|I_{u_\ell}|-1}$ such submonomials $m_\ell$ by Remark~\ref{rem:isomorphic-subpolynomials}.
	It follows that the number of monomials of the polynomial computed in node $\alpha$ is at most
	\[
	\prod_{\ell=1}^p r^{\lvert I_{u_{\ell}} \rvert -1} \leq r^{\lvert I \rvert -1} = r^{d_\alpha-1}.\qedhere
	\]
\end{proof}	

We are ready to prove hardness of the polynomial $H$ for monotone computation.
We shall make use of the following ``product lemma'', which comes in different forms in e.g.~\cite{SY, HrubesY11,Ramprasadsurvey}.

\begin{lemma}\label{lem:product-decomposition}
	A degree-$d$ homogeneous formula $F$ of size $s$ and product-depth $\Delta$
	can be written as a sum of $O(s)$ polynomials, 
	each of which
	is a product of $\Omega(\Delta d^{1/\Delta})$ many polynomials of
	positive degree. Moreover, each of these polynomials is computed by
	some gate in $F$.
\end{lemma}

\begin{proposition}
	\label{prop:hardness-of-H}
	If \(F\) is a monotone formula of product-depth \(\Delta  \leq \log d\)
	which computes \(H\), then its size is at least $r^{\Omega(\Delta d^{1/\Delta})}$.
\end{proposition}

\begin{proof}
	Let $t=\Delta d^{1/\Delta}$.
	Since $F$ is monotone, it is homogeneous and by Lemma~\ref{lem:product-decomposition} can be written as
	$$F = \sum_{i=1}^{s'} \prod_{j=1}^{t_i} F_{i,j}$$
	with $s'=O(s)$ and $t_i =\Omega(t)$ for all $i$, 
	and $F_{i,j}$ is of degree at least $1$ and computed by some gate in $F$.
	
	Let $d_{i,j}$ be the degree of $F_{i,j}$.
	By Lemma~\ref{lem:number-of-monomials}, each $F_{i,j}$ computes at most $r^{d_{i,j}-1}$ monomials.
	The number of monomials of $\prod_{i=1}^{t_j} F_{i,j}$ is therefore bounded by
	$$\prod_{j=1}^{t_i} r^{d_{i,j}-1} \leq r^{\sum_{j=1}^{t_i} d_{i,j}-t_i} \leq r^{d-t}$$
	since $\sum_{j=1}^{t_i} d_{i,j} = d$.
	It follows that the number of monomials computed by $F$ is at most $s \cdot r^{d-t}$.
	Since $F$ computes $H$ which has $r^{d-1}$ monomials, we get $s = r^{\Omega (t)}$.
\end{proof}

We can now get Theorem~\ref{thm:lowerbound-intro} from the introduction, which is restated here for convenience.
\begin{theorem}
	\label{thm:lowerbound}
	Let $n$ and $d = d(n)$ be growing parameters such that $d(n)\leq \sqrt{n}$. Then there is a monotone algebraic formula $F$ of size at most $n$ and depth $O(\log d)$ computing a polynomial $P\in \F[x_1,\ldots,x_n]$ of degree at most $d$ such that any monotone formula $F'$ of depth $o(\log d)$ computing $P$ must have size $n^{\omega(1)}.$
\end{theorem}

\begin{proof}[Proof of Theorem~\ref{thm:lowerbound-intro}]
	Choose parameters $k(n)$ and $r(n)$ 
	such that $P(n):=H^{(k,r)}$ has $\Theta(n)$ variables and degree $\Theta(d)$: let
	$k=\log d$ and $r=\frac{1}{2} n^{1/\log d}$.
	Condition $d(n)\leq \sqrt{n}$ ensures that $r\geqslant 2$.
	The polynomial $P$ has a monotone formula of size $O(n)$ and depth $O(\log d)$.
	By Proposition~\ref{prop:hardness-of-H}, 
	any monotone formula of product-depth $\Delta \leq \log d$ computing $P$
	has size
	$$r^{\Omega(\Delta d^{1/\Delta})} = 
	\left(\frac{1}{2} n^{1/\log d}\right)^{\Omega(\Delta d^{1/\Delta})} 
	= (n/d)^{\left( \frac{\Delta d^{1/\Delta}}{\log d} \right)}$$
	which is $n^{\Omega \left( \frac{\Delta d^{1/\Delta}}{\log d} \right)}$
	using the hypothesis $d(n)\leq \sqrt{n}$.
	Since $\frac{\Delta d^{1/\Delta}}{\log d} \rightarrow +\infty$ when $\Delta = o \left(\log d \right)$ this bound is $n^{\omega(1)}.$
\end{proof}

\section{Conclusion and Open questions}
\label{sec:final}

In this paper we investigated the possibility of reducing the depth of a formula of size \(s\) computing a polynomial of degree \(d\) to \(O(\log d)\) while keeping the size \(s^{O(1)}\).

We showed (Theorem~\ref{thm:main-intro}) that we can do such a transformation when \(F\) is {\em homogeneous}. More generally, Theorem~\ref{thm:main-syntactic-degree} states that we can achieve it as soon as the syntactic degree of \(F\) is polynomially bounded in \(d\).

\paragraph*{Structure inside VF.} Let us consider  a sequence of polynomials $(f_n)$ whose number of variables and  degree are bounded polynomially in \(n\) (such a family is usually called a \(p\)-family, see for example~\cite{BurgisserVPVNP}). We can then consider three classes of such families:
\begin{itemize}
	\item \(\text{homF}[s(n)] = \{ (f_n) \mid  f_n \text{ is computed by a homogeneous formula of size }\poly(s(n))\}\), 
	\item \(\text{lowSynDegF}[s(n)] = \{ (f_n) \mid f_n \text{ is computed by a formula of size }\poly(s(n)) \text{ and of} \\ \text{ syntactic degree } \poly(\deg(f_n))\}\)
	\item \(\text{lowDepthF}[s(n)] = \{ (f_n) \mid f_n \text{ is computed by a formula of size }\poly(s(n)) \text{ and of} \\ \text{ depth } O(\log \deg(f_n))\}\).
\end{itemize}

Clearly, we have the inclusion \(\text{homF}[s(n)] \subseteq \text{lowSynDegF}[s(n)]\). Also, in this paper, we have shown the inclusion \(\text{lowSynDegF}[s(n)] \subseteq \text{lowDepthF}[s(n)]\). Consequently, 
\[\text{homF}[\poly(n)] \subseteq \text{lowSynDegF}[\poly(n)]\subseteq \text{lowDepthF}[\poly(n)] \subseteq \text{VF},
\]
and we do not know if these inclusions are strict or not. 

\paragraph{The complexity of the Elementary Symmetric Polynomials.} A particularly interesting special case of the questions above comes from the example of the \emph{Elementary Symmetric Polynomials.} Given parameters $d, n$ with $d \leq n,$ recall that the Elementary Symmetric polynomial $S_n^d(x_1,\ldots, x_n)$ is the sum of all the multilinear monomials of degree exactly $d$. A simple and elegant construction of Ben-Or (see~\cite{SW}) shows that for any $d$, the polynomial $S_n^d$ has an \emph{inhomogeneous} formula of depth-$3$ and size $O(n^2).$ This puts this family of polynomials in the class $\text{lowDepthF}[\poly(n)].$ Further, Shpilka and Wigderson~\cite[Theorem 5.3]{SW}, showed that $S_n^d$ has depth-$6$ formulas of syntactic degree at most $\poly(d),$ putting it in the class $\text{lowSynDegF}[\poly(n)].$

However, as far as we know, there are no known $\poly(n)$-sized homogeneous formulas for this family of polynomials.\footnote{A strong form of this was conjectured by Nisan and Wigderson~\cite{NW95}, which was subsequently refuted by Hrube\v{s} and Yehudayoff~\cite{HrubesY11}. However, this still does not yield polynomial-sized homogeneous formulas for all elementary symmetric polynomials.} In fact, under some further restrictions, Hrube\v{s} and Yehudayoff showed~\cite{HrubesY11} a superpolynomial homogeneous formula lower bound when $d = n/2$. Removing these restrictions would show a separation between $\text{homF}[\poly(n)]$ and $\text{lowSynDegF}[\poly(n)]$. On the other hand, if indeed the elementary symmetric polynomials have $\poly(n)$-sized homogeneous formulas, then this can be used to argue\footnote{see e.g. \cite[Section III]{LST-FOCS} for the standard argument} the same for any polynomial computed by a depth-$3$ formula of polynomial size, hinting at a possible collapse between $\text{homF}[\poly(n)]$ and $\text{lowSynDegF}[\poly(n)]$.

\paragraph*{Lower bounds for higher-depth formulas.} Due to the recent lower bound results of~\cite{LST-FOCS}, we know that there is an explicit homogeneous polynomial $P(X)$ of degree $d$ on $n$ variables that cannot be computed by any formula of size $\poly(n)$ and depth $\varepsilon\cdot \log \log d$, for some absolute constant $\varepsilon > 0$. It turns out that the polynomial $P$ \emph{is} computable by an \emph{algebraic branching program} and therefore, lies in the complexity class called VBP. 

It is known that VF is contained in VBP. However, we do not know whether this containment is strict or not. Our lower bound result helps us pose a refined version of this question. Specifically, it shows that if the lower bound from~\cite{LST-FOCS} can be improved from $\Omega(\log \log d)$ to $\omega(\log d)$, then we will have separated VF from VBP. 

The fact that our depth reduction carries over to the non-commutative setting, makes a compelling case for revisiting the VF vs. VBP question in the non-commutative setting. Specifically, a recent result of~\cite{LST-STOC} shows that there is an explicit non-commutative polynomial $P(X)$ of degree $d$ on $n$ variables  that cannot be computed by any non-commutative formula of size $\poly(n)$ and depth $\varepsilon\cdot \sqrt{\log d}$. So, improving the lower bound in this case from $\Omega(\sqrt{\log d})$ to $\omega(\log d)$ would separate VF from VBP in the non-commutative setting.

\bibliographystyle{alpha}
\bibliography{references}

\clearpage
\appendix

\section{The Bonet-Buss depth-reduction}\label{sec:BonetBuss}

In this section, we give a proof of Theorem~\ref{thm:BCE}. This is essentially the same proof as in the work of Bonet and Buss~\cite{BonetBuss}, but since the results in that paper are only stated for Boolean formulas, we give the proof here for completeness.

We show that for some large enough absolute constant $C > 0$, any polynomial $P$ computed by a formula $F$ of size $s$ can also be computed by a formula $F'$ of size at most $s^{1+\varepsilon}$, depth at most $2^{O(1/\varepsilon)}\cdot \log s$, and fan-in $2$. Further, if $F$ is homogeneous/monotone/non-commutative, then so is $F'$. This latter claim will follow directly from the construction.

We assume here that each gate in $F$ has fan-in $2$ to begin with. We can make this modification to $F$ in the beginning at the expense of increasing the depth, and without increasing the size.

We prove the claim by induction on the size $s$ of the formula. Let $T(s)$ and $D(s)$ denote the maximum size and depth (respectively) of $F'$ thus obtained, assuming that $F$ has size at most $s.$ Let $k = 2^{C/\varepsilon}$ where $C > 0$ is an absolute constant we will choose below. We assume throughout that $\varepsilon < 1$, which is without loss of generality.

The base case of the formula is when $s \leq k,$ in which case the claim is trivial, as any formula of size $s$ has depth at most $s$, which in this case is at most $2^{C/\varepsilon}.$

Now, assume that $s > k.$ In this case, we find a gate $\alpha$ 
such that the size of the subformula $F_\alpha$ rooted at $\alpha$ has size at least $s-s/k$, while the children $\beta$ and $\gamma$ of $\alpha$ do not satisfy this property. It is easy to observe that there is a unique $\alpha$ with this property. 
Let $*$ denote the operation (either $+$ or $\times$) labelling $\alpha.$

We replace $\alpha$ by a fresh variable $y$ in $F$ to get a formula $F_y$. Note that $F_y$ has size at most $s_1 =  \frac{2s}{k}.$ Since $F_y$ has at most one occurrence of $y$, we can write
\[
F_y = A \cdot y \cdot B + C 
\]
where $A$ is the product of all subformulas that multiply $y$ on the left in $F$ along the path to the output, $B$ is similarly the product of all subformulas that multiply $y$ on the right, and $C$ is the polynomial computed by $F_y$ when we set $y$ to $0$. Clearly, $A$, $B$ and $C$ have formulas $F_A, F_B$ and $F_C$ of size at most $s_1$. Finally, we get the formula $F'$ as in Figure~\ref{fig}.
\begin{figure}[h]
	\centering
	\tikzstyle{gate}=[circle,draw=black!50,thick]
	\tikzstyle{leaf}=[minimum size = 20pt]
	\begin{tikzpicture}
		\node[gate, minimum size = 5pt] (output) at (0,0) {$+$};
		\node[gate, minimum size = 5pt] (times1) at (-0.5,-1) {$\times$}
		edge[->] (output);
		\node[inner sep = 0pt] (C) at (0.5,-1) {$F_C'$}
		edge[->] (output);
		\node[gate, minimum size = 5pt] (times2) at (-1,-2) {$\times$}
		edge[->] (times1);
		\node[inner sep = 0pt] (B) at (0,-2) {$F_B'$}
		edge[->] (times1);
		\node[inner sep = 0pt] (A) at (-1.5,-3) {$F_A'$}
		edge[->] (times2);
		\node[gate, minimum size = 5pt] (star) at (-0.5,-3) {$*$}
		edge[->] (times2);
		\node[inner sep = 0pt] (beta) at (-1,-4) {$F_\beta'$}
		edge[->] (star);
		\node[inner sep = 0pt] (gamma) at (0,-4) {$F_\gamma'$}
		edge[->] (star);
	\end{tikzpicture}
	\caption{Constructing $F'$}
	\label{fig}
\end{figure}
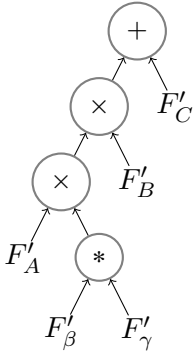

Here, $F_A', F_B', F_C', F_\beta', F_\gamma'$ are the formulas obtained by recursively applying the same procedure to $F_A, F_B, F_C, F_\beta, F_\gamma$.

We can bound the size and depth of the depth-reduced formula by induction. We have
\[
D(s) \leq D(s-s/k) + O(1)
\]
leading easily to an overall depth bound of $O(k\log s)  = 2^{O(1/\varepsilon)}\cdot \log s$ for any choice of the constant $C$. 

For the size, we have the following recursion.
\[
T(s) \leq \max_{\substack{s_\beta,s_\gamma, s_1:\\ s_\beta, s_\gamma \leq s(1-1/k)\\ s_1 \leq 2s/k\\ s_\beta+s_\gamma+ s_1 \leq s}}T(s_\beta) + T(s_\gamma) + 3T(s_1)
\] 
where $s_\beta$ and $s_\gamma$ are the sizes of $F_\beta$ and $F_\gamma$ respectively. We now use induction to bound $T(s)$ as follows. (We omit the conditions on $s_\beta, s_\gamma$ and $s_1$ for notational simplicity.)
\begin{align*}
	T(s) &\leq \max_{s_\beta,s_\gamma, s_1} s_\beta^{1+\varepsilon} + s_\gamma^{1+\varepsilon} + 3 s_1^{1+\varepsilon}\\
	&\leq \left(s\left(1-\frac{1}{k}\right)\right)^{1+\varepsilon} + \left(\frac{s}{k}\right)^{1+\varepsilon} + 3 \left(\frac{2s}{k}\right)^{1+\varepsilon}\\
\end{align*}
where we used the fact that $s_1 \leq 2s/k$ and the fact that, using the convexity of the map $x\mapsto x^{1+\varepsilon},$ the function $s_\beta^{1+\varepsilon} + s_\gamma^{1+\varepsilon}$ is maximized when $\max\{s_\beta,s_\gamma\} = s-s/k$, meaning that $\min\{s_\beta,s_\gamma\}\leq s/k.$ 

Continuing the computation, we get
\begin{align*}
	T(s) \leq s^{1+\varepsilon} \left(\left(1-\frac{1}{k}\right) + \frac{1}{k^{1+\varepsilon}} + 3\frac{4}{k^{1+\varepsilon}}\right)\leq s^{1+\varepsilon}\cdot \left( 1-\frac{1}{k} + \frac{C'}{ k^{1+\varepsilon}}\right)
\end{align*}
for some large enough absolute constant $C'$. Setting $k = 2^{C/\varepsilon}$ for a large enough absolute constant $C$ gives us $T(s) \leq s^{1+\varepsilon}$, proving the inductive claim.

\end{document}